%% file: main.tex
\documentclass[a4paper,11pt]{article}
\usepackage[margin=1in]{geometry}
\usepackage[utf8]{inputenc}
\usepackage[small]{caption}
\usepackage{hyperref}
\usepackage{graphicx}
\usepackage{amsmath}
\usepackage{amsthm,amssymb,amsfonts}
\usepackage{bm}
\usepackage{xspace}
\usepackage{xparse}
\usepackage{nicefrac}
\usepackage{multirow}

\input{./macros.tex}

\title{\textbf{Biased Opinion Dynamics:\\When the Devil Is in the Details}}

\author{
Aris Anagnostopoulos\\
		{\small{}Sapienza Università di Roma}\\
		{\small{}Rome, Italy}\\
		{\small{}\texttt{aris@diag.uniroma1.it}}\\
\and Luca Becchetti\\
		{\small{}Sapienza Università di Roma}\\
		{\small{}Rome, Italy}\\
		{\small{}\texttt{becchetti@diag.uniroma1.it}}\\
\and Emilio Cruciani \\
		{\small{}Inria, I3S Lab, UCA, CNRS}\\
		{\small{}Sophia Antipolis, France}\\
		{\small{}\texttt{emilio.cruciani@inria.fr}}\\
\and Francesco Pasquale\\ 
		{\small{}Università di Roma Tor Vergata}\\
		{\small{}Rome, Italy}\\
		{\small{}\texttt{pasquale@mat.uniroma2.it}} 
\and Sara Rizzo\\
		{\small{}Gran Sasso Science Institute}\\
		{\small{}L'Aquila, Italy}\\
		{\small{}\texttt{sara.rizzo@gssi.it}}
}
\date{}

\begin{document}

\maketitle

\begin{abstract}
\input{./abstract.tex}
\end{abstract}

\vfill
{\footnotesize Partially supported by the ERC Advanced Grant 788893 AMDROMA ``Algorithmic and Mechanism Design Research in Online Markets'' 
and MIUR PRIN project ALGADIMAR ``Algorithms, Games, and Digital Markets.''}

\newpage
\input{./trunk/intro.tex}

\input{./trunk/model.tex}

\input{./trunk/majority.tex}

\input{./trunk/voter.tex}

\input{./trunk/conclusions.tex}

\newpage
\bibliographystyle{alpha}
\bibliography{esb}

\end{document}

%% file: macros.tex
\usepackage{xcolor}

\newtheorem{definition}{Definition}[section]
\newtheorem{lemma}[definition]{Lemma}
\newtheorem{theorem}[definition]{Theorem}

\newcommand{\Prob}[2]{\mathbf{P}_{#1} \left( #2 \right)}
\newcommand{\Expec}[2]{\mathbf{E}_{#1} \left[ #2 \right]}

\newcommand{\bigO}{\mathcal{O}}

\newcommand{\bx}{\mathbf{x}}

\newcommand{\bz}{\mathbf{z}}

\newcommand{\bone}{\mathbf{1}}
\newcommand{\bzero}{\mathbf{0}}

\renewcommand{\le}{\leqslant}
\renewcommand{\leq}{\leqslant}
\renewcommand{\ge}{\geqslant}
\renewcommand{\geq}{\geqslant}

\newcommand{\better}{dominant}
\newcommand{\HCgraph}{Cycle Binary Configuration Graph}

%% file: abstract.tex
We investigate opinion dynamics in multi-agent networks when a bias toward
one of two possible opinions exists; for example, reflecting a status quo vs a
superior alternative.

Starting with all agents sharing an initial opinion representing the status
quo, the system evolves in steps. In each step, one agent selected uniformly at
random adopts the superior opinion with some probability $\alpha$, and with
probability $1 - \alpha$ it follows an underlying update rule to revise its
opinion on the basis of those held by its neighbors.
We analyze convergence of the resulting process under two well-known update 
rules, namely \textit{majority} and \textit{voter}.

The framework we propose exhibits a rich structure, with a non-obvious
interplay between topology and underlying update rule.
For example, for the voter rule we show that the speed of convergence 
bears no significant dependence on the underlying topology, 
whereas the picture changes completely under the majority rule, 
where network density negatively affects convergence.

We believe that the model we propose is at the same time simple, rich, and modular,
affording mathematical characterization of the interplay between bias,
underlying opinion dynamics, and social structure in a unified setting. 

%% file: trunk/intro.tex
\section{Introduction}
Opinion formation in social groups has been the focus of extensive research.
While many models considered in the literature confer the same {\em intrinsic}
value to all opinions~\cite{coates2018unified}, one might expect a group to
quickly reach consensus on a clearly ``superior'' alternative, if present. Yet,
phenomena such as \emph{groupthink} may delay or even prevent such an outcome.

In this perspective, we investigate models of opinion formation in which a bias
towards one of two possible opinions exists, for instance, reflecting intrinsic
superiority of one alternative over the other.\footnote{Characterizing the
notion of ``superiority'' is typically context-dependent and may be far from
obvious. We remark that this aspect is outside the scope of this paper.} In the
remainder, we use labels $0$ and $1$ for the two opinions and we assume $1$ is
the {\em \better} opinion, that is, the one towards which the agents have a
bias.  We investigate this question in a mathematically tractable setting,
informally described as follows.

\smallskip Assume some underlying opinion dynamics $\mathcal D$.  Starting from
an initial state in which all agents share opinion $0$, the system evolves in
rounds. In each round, one agent is selected uniformly at random. With some
probability $\alpha$, the agent adopts $1$, while with probability $1 -
\alpha$, the agent follows $\mathcal D$ to revise its opinion on the basis of
those held by its neighbors in an underlying network. 

\smallskip Although the general model we consider is simple and, under mild
conditions on $\mathcal D$, the family of processes it describes always admits
global adoption of opinion $1$ as the only absorbing state, convergence to this
absorbing state exhibits a rich variety of behaviors, which in non-obvious ways
depends on the interplay between the network structure and the underlying
opinion dynamics.  The relatively simple, yet general, model we consider allows
analytical investigation of the following question:
\begin{quote}
How does a particular combination of network structure and opinion dynamics
affects convergence to global adoption of the dominant opinion?  In particular,
how conducive is a particular combination to rapid adoption?
\end{quote}

\paragraph{Main findings.} In general, the interplay between underlying network
structure and opinion dynamics may elicit quite different collective behaviors.

In Section~\ref{sec:majority}, we show that the expected time for consensus on
the dominant opinion grows exponentially with the minimum degree under the
\emph{majority} update rule, in which agents update their opinion to the
majority opinion in their neighborhoods~\cite{krapivsky2003dynamics}.  Using
asymptotic notation and denoting the number of agents in the network by $n$, we
obtain that convergence time is super-polynomial in expectation whenever the
minimum degree is $\omega(\log n)$. One might wonder, if the converse occurs,
namely, if a logarithmic maximum degree affords (expected) polynomial
convergence to the absorbing state. Even though we prove that this is indeed
the case for specific topologies as cycles or restricted graph families, this
does not seem to hold in general (see discussion in Section~\ref{se:concl}).

The results for majority are at odds with those we obtain in
Section~\ref{sec:voter} for the \emph{voter model}, where agents copy the
opinion of a randomly selected neighbor~\cite{liggett2012interacting}. In this
case, convergence to the absorbing state occurs within
$\mathcal{O}\left(\frac{1}{\alpha}n\log n\right)$ rounds {\em with high
probability}, regardless of the underlying network structure. We emphasize that
convergence time remains $\mathcal{O}\left(n^{1+s}\log n\right)$ when $\alpha =
\Theta\left(\frac{1}{n^s}\right)$ for any $s > 0$. 

Although results suggesting a negative impact of network density on convergence
time have been proposed in the past, albeit for quite different models
(e.g.,~\cite{montanari2010spread}), the results above suggest that there might
be more to the issue. In particular, the interplay between opinion dynamics
and underlying network structure seems more complex than anticipated, with the
former playing a key role in amplifying network effects. 

\smallskip At a higher level, we provide a simple mathematical framework to
investigate the interplay between opinion dynamics and underlying network
structure in a unified setting, allowing comparison of different update rules
with respect to a common framework. In this respect, we hope that our work
moves in the direction of a shared framework to investigate opinion dynamics,
as advocated in~\cite{coates2018unified}. 

\subsection{Related Work}\label{sec:related}
The problem we consider touches a number of areas where similar settings have
been considered, with various motivations. The corresponding literature is vast
and providing an exhaustive review is unfeasible here. In the paragraphs that
follow, we discuss contributions that most closely relate to the topic of this
paper. 

\paragraph*{Opinion diffusion and consensus.} Opinion dynamics are widely used
to investigate how group of agents modify their beliefs under the influence of
other agents and possibly exogenous factors. A number of models have been
proposed in the more or less recent past, mostly motivated by phenomena that
arise in several areas, ranging from social sciences, to physics and biology.
The reader is referred to~\cite{coates2018unified} and references therein for a
recent, general overview of opinion dynamics in multi-agent systems.  A first
distinction is between settings in which the set of possible beliefs is
continuos, e.g., the interval $[0, 1]$. This setting has been the focus of
extensive research in social sciences and
economics~\cite{degroot1974reaching,friedkin1990social,friedkin2017truth}. 

In this paper, we consider the case in which opinions are drawn from a discrete
set, a setting that also received significant attention in the recent past. In
particular, we focus on the {\em majority rule} and the {\em voter model}.
Investigation of the majority update rule originates from the study of
agreement phenomena in spin systems~\cite{krapivsky2003dynamics}, while the
voter model was motivated by the study of spatial conflict between species in
biology and interacting stochastic processes/particle systems in probability
theory and
statistics~\cite{clifford1973model,holley1975ergodic,liggett2012interacting}.
These two models received renewed attention in the recent past, the focus
mostly being on the time to achieve consensus and/or conditions under which
consensus on one of the initial opinions is achieved with a minimum degree of
confidence.  The voter model is by now well understood. In particular,
increasingly tight bounds on convergence time for general and specific
topologies have been proposed over the recent
past~\cite{hassin1999distributed,cooper2013coalescing}, while it is known that
the probability of one particular opinion to prevail is proportional to the sum
of the degrees of nodes holding that opinion at the onset of the
process~\cite{donnelly1983finite}. 

\paragraph{Consensus and network structure.} Network structure has been known
to play an important role in opinion diffusion and influence spreading for
quite some time~\cite{morris2000contagion}, under a variety of models.  For
example, consensus under the voter model and dependence of its convergence on
the underlying network topology have been thoroughly
investigated~\cite{donnelly1983finite,hassin1999distributed,cooper2013coalescing}.
For majority dynamics,~\cite{auletta2015minority} characterized topologies for
which an initial majority can be subverted, showing that this is possible for
all but a handful of topologies, including cliques and quasi-cliques. On the
other hand, regardless of the network, there is always an initial opinion
distribution, such that the final majority will reflect the initial one, while
computing an initial opinion configuration that will subvert an initial
majority is topology-dependent and NP-hard in
general~\cite{auletta2018reasoning}.

A number of recent contributions investigated (among other aspects) the
relationship between network structure and consensus in opinion formation
games~\cite{ferraioli2016decentralized,ferraioli2017social}, while extensions
of the Friedkin-Johnsen model to evolving networks were investigated
in~\cite{auletta2019consensus}.

While expansion of the underlying graph typically accelerates
convergence~\cite{cooper2012coalescing,kanade2019coalescence} in many opinion
dynamics, some recent work explicitly points to potentially adverse effects of
network structure on the spread of innovation, at least in scenarios where
opinion update occurs on the basis of private utilities that reflect both the
degree of local consensus and intrinsic value of the competing
opinions~\cite{montanari2010spread,young2011dynamics}. 

While some of our findings are qualitatively consistent with previous work
albeit under completely different models (in
particular,~\cite{montanari2010spread}), our overall approach is very
different, since it completely decouples the mechanism of opinion formation
from modelling of the bias, affording a clear-cut mathematical characterization
of the interplay between bias, underlying opinion dynamics and network
structure.

\paragraph*{Different forms of bias.} Bias in opinion dynamics has been
considered previously in the literature. We briefly review contributions that
are at least loosely related to our framework. For the voter and majority
update rules,~\cite{mukhopadhyay2016binary} introduces bias in the form of
different, opinion-dependent firing rate frequencies of the Poisson clocks that
trigger agents' opinion updates, implicitly enforcing a bias toward the opinion
with lower associated rate. While different, their model is similar to ours in
spirit and some of their results for the voter model are consistent with ours.
Yet, these results only apply in expectation and to very dense networks with
degree $\Omega(n)$, whereas our results for the voter model hold for every
undirected graph.

A somewhat related line of research addresses the presence of stubborn agents
or zealots. Loosely speaking, stubborn agents have a bias toward some
(initially or currently) held opinion, while zealots are agents that never
deflect from some initial opinion. Restricting to the discrete-opinion setting,
which is the focus of this paper,\footnote{For the continuous case, there is a
vast literature; see the seminal paper~\cite{friedkin1990social} and follow-up
work.} the role of zealots and their ability to subvert an initial majority
have been investigated for the voter model (see~\cite{mobilia2003does} and
follow-up work), while~\cite{auletta2017information} investigates majority
dynamics in the presence of stubborn agents that are biased toward the
currently held opinion, providing a full characterization of conditions under
which an initial majority can be subverted.

%% file: trunk/model.tex
\section{Notation and Preliminaries}\label{sec:prel}
Let $G = (V,E)$ be an undirected graph with $|V|=n$ nodes, each representing an
agent. Without loss of generality, we assume that $V = [n]:=\{1,\dots,n\}$.
The system evolves in discrete time steps\footnote{This is equivalent to the
asynchronous model in which a node revises its opinion at the arrival of an
independent Poisson clock with rate $1$ \cite{boyd2006randomized}.} and, at any
given time $t \in \mathbb{N}$, each node $v \in V$ holds an {\em opinion}
$x_v^{(t)} \in \{0,1\}$. We use the term opinion liberally here, in the sense
that $0$ and $1$ in general represent competing alternatives, whose meaning is
context-dependent and outside the scope of this paper.  We denote by $\bx^{(t)}
= \big( x_1^{(t)},\ldots,x_n^{(t)} \big)^\intercal$ the corresponding {\em
state} of the system at time $t$. We assume that the initial state of the
system is $\bx^{(0)} = \bzero = (0, \dots, 0)^\intercal$; such assumption is
discussed in Section \ref{se:concl}.  For each $v\in V$, we denote the
neighborhood of $v$ with $N_v := \{u \in V : \{u,v\} \in E\}$ and the degree of
$v$ with $d_v := |N_v|$.  Finally, $\Delta := \min_{v \in V} d_v$ is the
minimum degree of the nodes in $G$. 

Our framework assumes that agents exhibit a bias toward one of the opinions
(e.g., reflecting intrinsic superiority of a technological innovation over the
{\em status quo}), without loss of generality $1$, which we henceforth call the
{\em \better\ opinion}. We model bias as a probability, with a parameter
$\alpha \in (0,1]$. All dynamics we consider are \textit{Markovian}, that is,
given the underlying graph $G$, the distribution of the state $\bx^{(t)}$ at
round $t$ only depends on the state $\bx^{(t-1)}$ at the end of the previous
round. Moreover, they have $\bx = \bone = (1, \dots, 1)^\intercal$ as the only
absorbing state. We use $\tau$ to denote the \emph{absorption time}, which is
the number of rounds for the process to reach the absorbing state $\bone$.
Finally, for a family of events $\{\mathcal{E}_n\}_{n \in \mathbb{N}}$ we say
that $\mathcal{E}_n$ occurs \emph{with high probability} (\emph{w.h.p.}, in
short) if a constant $\gamma > 0$ exists such that $\Prob{}{\mathcal{E}_n} = 1
- \mathcal{O}(n^{-\gamma})$, for every sufficiently large~$n$.

%% file: trunk/majority.tex
\section{Absorption Time for Majority Dynamics}\label{sec:majority}
In this section, we investigate the time to reach consensus on the dominant
opinion under the majority update rule. More formally, we study the following
random process: Starting from the initial state $\bx^{(0)} = \left(0,\ldots ,
0\right)^\intercal$, in each round $t$ a node $u \in [n]$ is chosen uniformly
at random and $u$ updates its value according to the rule
\[
x_u^{(t)} = 
\left\{
\begin{array}{ll}
1 \quad & \quad \text{with probability } \alpha,\\
M_G(u, \bx) & \quad \text{with probability } 1 - \alpha,
\end{array}
\right.
\] 
where $\alpha \in (0,1]$ is the bias toward the \better\ opinion $1$ and
$M_G(u, \bx)$ is the value held in configuration $\bx^{(t-1)} = \bx$ by the
majority of the neighbors of node $u$ in graph $G$:
\[
M_G(u, \bx) = 
\left\{
\begin{array}{ll}
0 & \quad \text{if } \sum_{v \in N_u} x_v < |N_u|/2,\\
1 & \quad \text{if } \sum_{v \in N_u} x_v > |N_u|/2,\\
\end{array}
\right.
\]
and ties are broken uniformly at random, that is, if $\sum_{v \in N_u} x_v =
|N_u|/2$ then $M_G(u, \bx) = 0$ or $1$ with probability $1/2$.  

It is straightforward to see that for every positive $\alpha$, the above Markov
chain has $\bone$ as the only absorbing state. However, the rate of convergence
is strongly influenced by the underlying graph $G$. In
Subsection~\ref{subse:high_dens} we prove a lower bound on the expected
absorption time that depends exponentially on the minimum degree. This result
implies super-polynomial expected absorption times for graphs whose minimum
degree is $\omega(\log n)$. On the other hand, in
Subsection~\ref{subse:low_dens} we prove that the absorption time is $\bigO(n
\log n)$ on cycle graphs, and further graph families with sub-logarithmic
maximum degree and polynomial (expected) absorption time are briefly discussed
in Subsection~\ref{subse:other}.

\subsection{Slow Convergence on High-Density Graphs}\label{subse:high_dens}
In this section we prove a general lower bound on the expected absorption time,
which only depends on the minimum degree~$\Delta$. To this purpose, we use the
following standard lemma on birth-and-death chains\footnote{Birth-and-death
chains are Markov processes for which, if in state $k$, a transition could only
go to either state $k+1$ or state $k-1$.} (see,
e.g.,~\cite[Section~17.3]{levin2017markov} for a proof).
\begin{lemma}\label{lemma:biasedbdchainsprob}
Let $\{X_t\}_t$ be a birth-and-death chain with state space $\{0,1, \dots, n\}$
such that for every $1 \leq k \leq n-1$
\begin{align*}
	\Prob{}{X_{t+1} = k+1 \,|\, X_t = k} & = p,
	\\
	\Prob{}{X_{t+1} = k-1 \,|\, X_t = k} & = q,
	\\
	\Prob{}{X_{t+1} = k \,|\, X_t = k} & = r,
\end{align*}
with $p + q + r = 1$. For every $i \in \{0,1, \dots, n\}$ let $\tau_i$ 
be the first time the chain hits state $i$, that is,
\(
	\tau_i = \inf \{ t \,|\, X_t = i \}.
\)
If $0 < p < 1/2$, the probability that starting from state $k$ the chain 
hits state $n$ before state $0$ is
\[
	\Prob{k}{\tau_n < \tau_0} = \frac{(q/p)^k - 1}{(q/p)^n - 1} 
	\leq \left(\frac{p}{q}\right)^{n-k}.
\]
\end{lemma}

It is not difficult to show that, for $\alpha \geqslant 1/2$, every graph with
minimum degree $\Delta = \Omega(\log n)$ has $\bigO(n \log n)$ absorption time,
w.h.p. Indeed, since every time a node updates its opinion the node chooses
opinion $1$ with probability at least $\alpha$, as soon as all nodes update
their opinion at least once (it happens within $\bigO(n \log n)$ time steps,
w.h.p., by coupon collector argument) if $\alpha \geqslant 1/2$, every node $u$
will have a majority of $1$s in its neighborhood, w.h.p. 

In the next theorem we prove that, as soon as $\alpha$ is smaller than $1/2$,
the absorption time instead becomes exponential in the minimum degree.

\begin{theorem}\label{thm:min_degree}
Let $G = (V, E)$ be an undirected graph with minimum degree $\Delta$. 
Assume $\alpha \leq \frac{(1-\varepsilon)}{2}$, for an
arbitrary constant $0 < \varepsilon < 1$. The expected absorption time for the
biased opinion dynamics under the majority update rule is 
\[
\Expec{}{\tau} \geq \frac{e^{\frac{\varepsilon^2}{6} \Delta}}{6n}.
\]
\end{theorem}
\begin{proof}
Let $S^{(t)}$ be the set of nodes with value $1$ at time $t$. 
For each node $u \in V$, let $n^{(t)}_u$ be the fraction of its neighbors with
value $1$ at round $t$:
\[
	n^{(t)}_u = \frac{| N_u \cap S^{(t)} |}{|N_u|}.
\]
Finally, let $\bar{\tau}$ be the first round in which
$n^{(t)}_u \geq 1/2$ for at least one node $v\in V$, namely,
\[
	\bar{\tau} = 
	\inf \left\{ 
	t \in \mathbb{N} \,:\, n^{(t)}_u \geq 1/2,
	\text{ for some } u \in [n]
	\right\}.
\]
Note that for each round $t \leq \bar{\tau}$ all nodes have a majority 
of neighbors sharing opinion $0$, thus the selected agent at time $t$ updates 
its state to 1 with probability $\alpha$ and to 0 with probability $1-\alpha$.
Moreover, clearly $\tau \geq \bar{\tau}$. We next prove that 
$\Expec{}{\bar{\tau}} \geq e^{\frac{\varepsilon^2}{6}
\Delta}/(6n)$, which implies our thesis.

Observe that, for a node $u$ with degree $d_u$ that has $k$ neighbors with
value $1$ in some round and for every $t \leq \bar{\tau}$, the probabilities
$p_k(u)$ and $q_k(u)$ of increasing and decreasing, respectively, of one unit
the number of its neighbors with value $1$ are
\[
p_k(u) = \frac{d_u - k}{n} \alpha,
\text{ and }
q_k(u) = \frac{k}{n} (1 - \alpha).
\]
Hence, because $\alpha \leq (1-\varepsilon)/2$, for every 
$k \geq d_u / (2+\varepsilon)$ we have that
\[
	\frac{p_k(u)}{q_k(u)} 
	= \frac{d_u - k}{k} \cdot \frac{\alpha}{1-\alpha}
	\leq (1+\varepsilon) \cdot \frac{1-\varepsilon}{1+\varepsilon}
	= 1 - \varepsilon.
\]
Note that
\[
	\frac{d_u}{2} - \frac{d_u}{2+\varepsilon}
	= d_u \frac{\varepsilon}{2(2+\varepsilon)}
	\geq \frac{\varepsilon}{6} d_u.
\]
From Lemma~\ref{lemma:biasedbdchainsprob} it thus follows that, for each node
$u$, as soon as the number of its neighbors with value $1$ enters in the range 
$\left( d_u/(2+\varepsilon), \, d_u /2 \right)$, the probability 
that it will reach $d_u/2$ before going back to  
$d_u/(2+\varepsilon)$ is at most
\[
	(1 - \varepsilon)^{\varepsilon d_u / 6} 
	\leq e^{- \varepsilon^2 d_u / 6}
	\leq e^{- \frac{\varepsilon^2}{6} \Delta},
\]
using $(1-x)^x \leq e^{-x^2}$ for $x \in [0, 1]$.  Hence, if we denote by $Y_u$
the random variable indicating the number of trials before having at least
$1/2$ of the neighbors of $u$ at $1$ we have that for every $t \geq 0$
\[
	\Prob{}{Y_u \geq t} 
	\geq \left( 1 - e^{-\frac{\varepsilon^2}{6} \Delta} \right)^t
	\geq e^{-(3t/2)e^{-\frac{\varepsilon^2}{6} \Delta}},
\]
where in the last inequality we used that $1 - x \geq e^{-3x/2}$ for every $x
\in [0, \frac{1}{2})$. Thus,
\[
\Prob{}{Y_u < t} 
\leq 1 - e^{-(3t/2)e^{-\frac{\varepsilon^2}{6} \Delta}}
\leq \frac{3t}{2} \, e^{-\frac{\varepsilon^2}{6} \Delta},
\]
using $1-e^{-x} \leq x$ for every $x$. Finally, by using the union bound over
all nodes, we have that
\[
\Prob{}{\bar{\tau} < t} 
= \Prob{}{\exists u \in [n] \,:\, Y_u < t}
\leq n \cdot \frac{3t}{2} e^{-\frac{\varepsilon^2}{6} \Delta}.
\]
Thus, for $\bar{t} = e^{\frac{\varepsilon^2}{6} \Delta}/3n$ we have
$\Prob{}{\bar{\tau} \leq \bar{t}} \leq 1/2$ and the thesis follows from Markov
inequality:
\[
\Expec{}{\bar{\tau}} 
\geq \bar{t} \, \Prob{}{\bar{\tau} \geq \bar{t}}
\geq \frac{\bar{t}}{2}.
\qedhere
\]
\end{proof}

\subsection{Fast Convergence on the Cycle}\label{subse:low_dens}
In this section, we prove that the absorption time on an $n$-node cycle graph
is $\bigO(\frac{1}{\alpha} n \log n)$, w.h.p. We make use of the following
\textit{structural} lemma.

\begin{lemma}[Structural property of cycles]\label{lemma:struct_cycle}
	Let $C_n$ be the cycle on $n$ nodes and let every node $v\in V$ 
	have an associated state $x_v \in \{0,1\}$. 
	Let us call $B_i$ and $S_i$ the set of nodes in state $i$ such that:
	every node $v \in B_i$ has both neighbors in the opposite state and 
	every node $v \in S_i$ has one single neighbor in the opposite state. 
	The following holds:
	\[
		|B_0| + \frac{|S_0|}{2} = |B_1| + \frac{|S_1|}{2}.
	\]
\end{lemma}
\begin{proof}
Given any possible binary coloring of $C_n$ each node $v$ belongs to one of the 
following categories:
\begin{itemize}
    \item $v \in B_i$: node $v$ is in state $i$ and both its neighbors are in
    state $j \neq i$.
    \item $v \in R_i$: node $v$ is in state $i$, its left neighbor is in state
    $i$, and its right neighbor is in state $j \neq i$.
    \item $v \in L_i$: node $v$ is in state $i$, its right neighbor is in state
    $i$, and its left neighbor is in state $j \neq i$.
    \item $v \in Z_i$: node $v$ is in state $i$ and zero of its neighbors are
    in state $j\neq i$, i.e., both are in state $i$.
\end{itemize}
We also call $S_i = R_i \cup L_i$. Figure~\ref{fig:legend} illustrates the eight 
(counting symmetries) possible categories.
\begin{figure}[ht]
	\centering
	\includegraphics[width=0.9\textwidth]{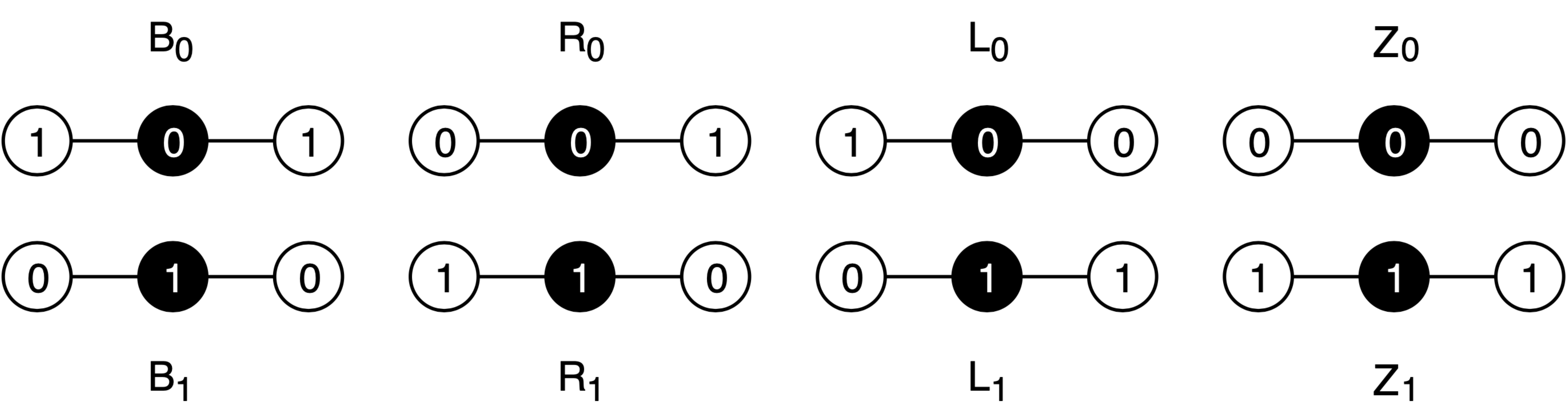}
    \caption{Categories of a node $v$ in $C_n$; node $v$ is black while its
    left and right neighbors are white.}
    \label{fig:legend}
\end{figure}
Let us consider a clockwise walk through $C_n$ that returns to its starting
point. Keeping into account the categories of the nodes previously described
it is possible to generate a graph $H_C$ that describes all possible binary
configurations of a $C_n$ graph, for every $n \in \mathbb{N}$.  We call $H_C$
the \emph{\HCgraph{}} (Figure~\ref{fig:cycle_conf_graph}).  The nodes of $H_C$
represent the possible categories of the nodes of $C_n$ while the edges the
possible neighbors in $C_n$, considering a clockwise walk.  For example, there
is no edge from $B_0$ to $R_0$ since the neighbors of $B_0$ are both in state
1, while a node in $R_0$ is in state 0.
\begin{figure}[ht]
	\centering
	\includegraphics[width=0.5\textwidth]{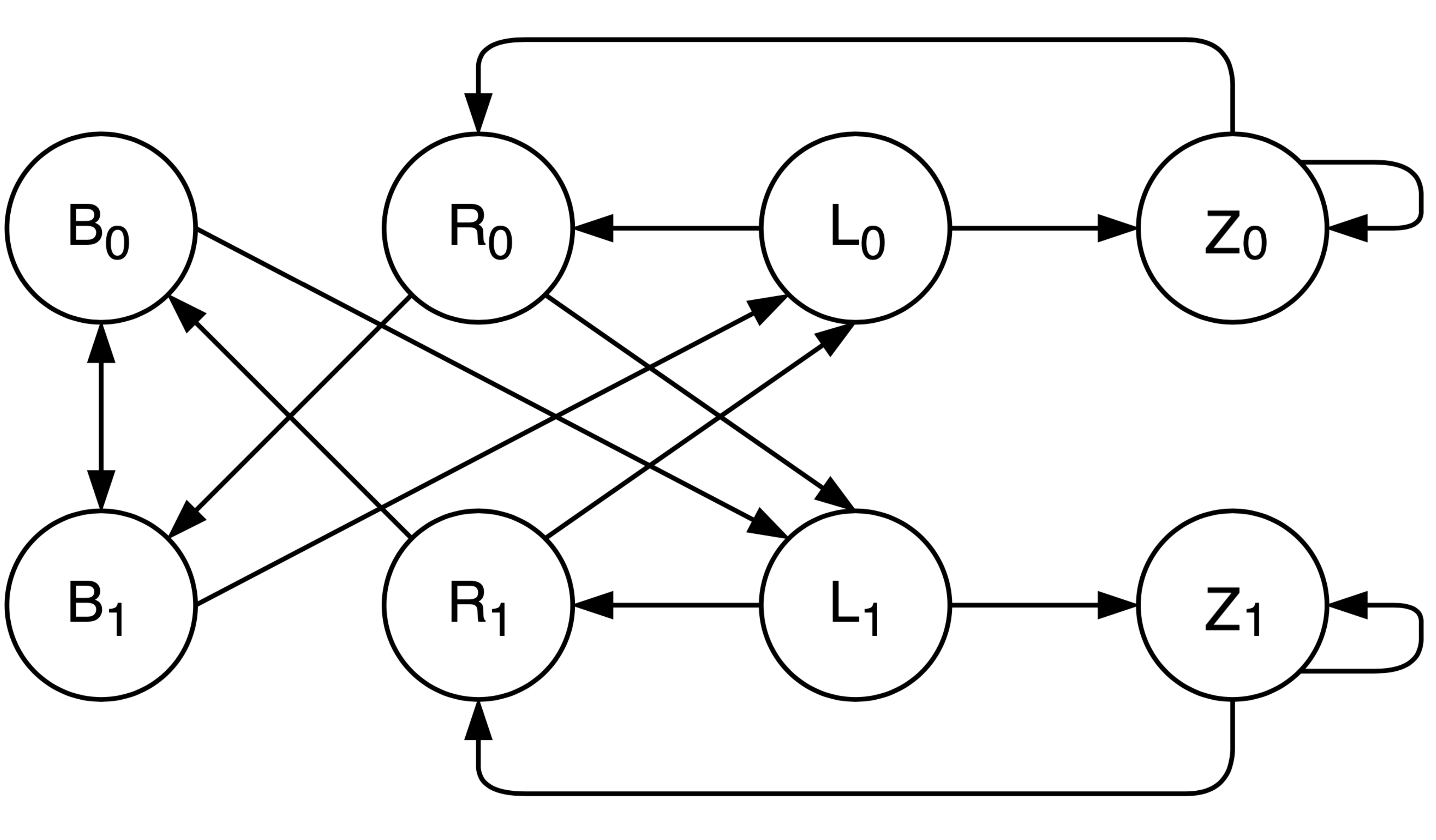}
	\caption{The \emph{\HCgraph{}} $H_C$.}
	\label{fig:cycle_conf_graph}
\end{figure}

Let us pick any node $v$ in $C_n$ and let us walk through clockwise until we
return to $v$. Let us pick the node of $H_C$ corresponding to the category $v$
belongs to and follow the clockwise walk that we do on $C_n$ also on $H_C$, by
moving on the corresponding states. It follows that after $n$ steps the walk on
$C_n$ will be back to $v$ and the walk on $H_C$ will be back to the node
representing the category of $v$.  Note that this implies that the walk on
$H_C$ is a cycle and, more in general, that every cycle of length $n$ on $H_C$
represent a possible binary configuration of the nodes of a corresponding cycle
graph $C_n$.

Note that every possible cycle in $H_C$ is a combination of simple cycles (that
go through each node at most once) on $H_C$.  We prove that the structural
property of the lemma holds for every simple cycle on $H_C$. By commutativity
and associativity of addition, the property directly transfer also to
composition of simple graphs.  In order to reduce the number of simple cycles
(which are 17; they are easy to find on a computer given the small size of the
graph $H_C$---the problem is
$\mathit{\#P}$-hard~\cite{arora2009computational}), we avoid cycles that pass
through $Z_i$ since $|Z_i|$ does not appear in the lemma; in fact, every
cycle passing through $Z_i$ does $L_i \rightarrow Z_i \rightarrow R_i$ and the
only other outgoing edge of $L_i$ is $L_i \rightarrow R_i$.  In other words,
excluding simple cycles passing through node $Z_i$ does not have any effect on
the following calculations.  By also taking advantage of symmetries in $i$ and
$j$, all the remaining simple cycles are the following four, for which the
equality of the lemma is true:
\begin{itemize}
	\item $(B_i \rightarrow B_j)$: 
		$|B_i|$ cancels out with $|B_j|$.
		\item $(B_i \rightarrow L_j \rightarrow R_j)$: 
		$|B_i|$ cancels out with $\frac{|R_j|+|L_j|}{2}$.
		\item $(R_i \rightarrow L_j \rightarrow R_j \rightarrow L_i)$:
		$\frac{|R_i|+|L_i|}{2}$ cancels out with $\frac{|R_j|+|L_j|}{2}$.
		\item $(B_i \rightarrow B_j \rightarrow L_i \rightarrow R_i \rightarrow L_j \rightarrow R_j)$: 
			$|B_i|$ cancels out with $|B_j|$; 
			$\frac{|R_i|+|L_i|}{2}$ cancels out with $\frac{|R_j|+|L_j|}{2}$.
			\qedhere
\end{itemize}
\end{proof}

\begin{theorem}[Cycles]\label{thm:cycle}
Let $G = C_n$ be the cycle on $n$ nodes. Under the majority update rule, we
have $\tau = \mathcal{O}\left(\frac{1}{\alpha} n \log n\right)$, with high
probability.
\end{theorem}
\begin{proof}
	Denote by $V_i$ the set of nodes with state $i$.
	Given a configuration $\bx\in\{0,1\}^n$ of $C_n$, 
	let $B_i = \{v \in V_i : \forall u \in N_v, x_u \neq i \}$
	and $S_i = \{v \in V_i : \exists u,w \in N_v, x_u \neq x_w \}$
	(see Lemma \ref{lemma:struct_cycle}).
	Let $X_t$ be the random variable indicating the number of nodes 
	in state $1$ at round $t$ and observe that for every $k$, we have:
	\[
		\Prob{}{X_t = h \,|\, X_{t-1} = k} = \left\{\begin{array}{ll}
			q_k & \quad \mbox{ if } h = k-1,
			\\
			r_k & \quad \mbox{ if } h = k,
			\\
			p_k & \quad \mbox{ if } h = k+1,
		\end{array}\right.
	\]
	where $q_k = (1-\alpha)\left(\frac{|B_1|}{n} +\frac{1}{2}\frac{|S_1|}{n}\right)$,
	$p_k = \alpha \frac{n-k}{n} + (1-\alpha)\left(\frac{|B_0|}{n} + \frac{1}{2}\frac{|S_0|}{n}\right)$,
	and $r_k = 1- q_k - p_k$. 
	Therefore, the expected value of $X_t$, conditioned to $X_{t-1} =k$, is
	\begin{align*}\label{eq:exp_cond}
	\Expec{}{X_t \;|\; X_{t-1} = k}
	&= (k-1)q_k + k r_k + (k+1)p_k 
	 = k - q_k + p_k\\
	&= k + \alpha \frac{n-k}{n} + 
	\frac{1-\alpha}{n}\left(|B_0| + \frac{|S_0|}{2} - |B_1| - \frac{|S_1|}{2}\right)\\
	&\stackrel{(a)}{=}k + \alpha \frac{n-k}{n},
	\end{align*}
	where in derivation $(a)$ we use Lemma~\ref{lemma:struct_cycle}.
	We therefore have:
	\begin{align*}
	\Expec{}{X_t}
	&= \sum_{k = 0}^n \Expec{}{X_t \;|\; X_{t-1} = k} 
	\Prob{}{X_{t-1} = k} 
	\\
	& = \sum_{k = 0}^n \left( \alpha + \left(1 - \frac{\alpha}{n}\right) k \right)
	\Prob{}{X_{t-1} = k} 
	\\
	& = \alpha \sum_{k = 0}^n \Prob{}{X_{t-1} = k} +
	\left(1 - \frac{\alpha}{n}\right) \sum_{k = 0}^n k \Prob{}{X_{t-1} = k} 
	\\
	& = \alpha + \left(1 - \frac{\alpha}{n} \right) \Expec{}{X_{t-1}}.
	\end{align*}
	Solving this recursion with $\Expec{}{X_0} = 0$ we get
	\[
	\Expec{}{X_t} 
	= \alpha \sum_{i = 0}^{t-1} \left( 1 - \frac{\alpha}{n} \right)^i 
	= \alpha \frac{1 - \left(1 - \alpha / n \right)^t}{\alpha / n}.
	\]
	The expected number $n - X_t$ of nodes in state $0$ at round $t$ is thus 
	\[
	\Expec{}{n - X_t} = n \left( 1 - \frac{\alpha}{n} \right)^t,
	\]
	that is smaller than $\frac{1}{n}$ for $t \geq \frac{2}{\alpha}n \log n$.
	Hence,
	\[
		\Prob{}{\tau > \frac{2}{\alpha} n \log n} 
		= \Prob{}{n - X_{\frac{2}{\alpha}n \log n} \geq 1}
		\stackrel{(b)}{\leqslant} \Expec{}{n - X_{\frac{2}{\alpha}n \log n}}
		\leqslant 1/n,
	\]
	where in $(b)$ we use the Markov inequality.
\end{proof}

\subsection{Further Low-Density Graph Families}\label{subse:other}
It is not difficult to show that convergence times are also polynomial in the
cases of trees of degree $\mathcal{O}(\log n)$ and disconnected cliques of size
$\bigO(\log n)$. These results are summarized as the following theorem.

\begin{theorem}[Trees and disconnected cliques]\label{thm:trees_cliques}
Assume $G = (V, E)$ is a tree of degree $\mathcal{O}(\log n)$ (resp.\ a set of
disconnected cliques, each of size $\mathcal{O}(\log n)$). Then, for every
constant $\alpha\in (0, 1]$, the expected absorption time is polynomial.
\end{theorem}

%% file: trunk/voter.tex
\section{Absorption Time for the Voter Model}\label{sec:voter}
As mentioned in the introduction, the voter model has received considerable
attention as an opinion dynamics in the more and less recent
past~\cite{liggett2012interacting}. It may be regarded as a ``linearized'' form
of the majority update rule, in the sense that, upon selection, a node pulls
each of the two available opinions with probability proportional to the
opinion's support within the node's neighborhood.  Despite such apparent
similarity, the two update rules result in quite different behaviors of the
biased opinion dynamics. Namely, for the voter model, absorption times to the
dominant opinion are polynomial with high probability as long as $1/\alpha$ is
polynomial, regardless of the underlying topology. These results are clearly at
odds with those of Section~\ref{sec:majority}. 

The biased voter model can formally be defined as follows: Starting from some
initial state $\bx^{(0)}$, at each round $t$ a node $u \in [n]$ is chosen
uniformly at random and its opinion is updated as
\[
x_u^{(t)} = 
\left\{
\begin{array}{ll}
1 \quad & \quad \text{with probability } \alpha,\\
V_G(u, \bx) & \quad \text{with probability } 1 - \alpha,
\end{array}
\right.
\] 
where $\alpha \in (0,1]$ is a parameter measuring the bias toward the better
opinion $1$ and $V_G(u, \bx)$ is the value held in configuration
$\bx^{(t-1)}=\bx$ by a node sampled uniformly at random from the neighborhood
of node $u$.  We assume $\bx^{(0)} = \bzero$ for simplicity, though we remark
that Theorem \ref{thm:conv_time_asy} below holds for any $\bx^{(0)}\in\{0,
1\}^n$.

As the proof of Theorem~\ref{thm:conv_time_asy} highlights, the biased opinion
dynamics under the voter update rule can be succinctly described by a {\em
nonhomogeneous} Markov chain~\cite{seneta2006non}. Although nontrivial to study
in general, we are able to provide tight bounds in probability for the
simplified setting we consider.

\begin{theorem}\label{thm:conv_time_asy}
Let $G=(V,E)$ be an arbitrary graph. The biased opinion dynamics with voter as
update rule reaches state $\bone$ within $\tau = \mathcal{O}(\frac{1}{\alpha} n
\log n)$ steps, with high probability.
\end{theorem}
\begin{proof}
For every node $v \in V$, the expected state of $v$ at time $t+1$, conditioned
on $\bx^{(t)} = \bx$ is
\begin{align*}
	\Expec{}{x_v^{(t+1)} \;|\; \bx^{(t)}=\bx} 
	&=\frac{1}{n} \left[\alpha + \frac{(1 - \alpha)}{d_v}\sum_{u \in N_v} x_u\right] +\left(1 - \frac{1}{n}\right)x_v
	\\
	&= \frac{\alpha}{n} + \frac{1}{n}\Big[ (1 - \alpha)(P\bx)_v + (n-1)(I\bx)_v\Big],
\end{align*}
where $P = D^{-1}A$ is the transition matrix of the simple random walk on $G$
(with $D$ the diagonal degree matrix and $A$ the adjacency matrix of the graph)
and $I$ is the identity matrix. Considering all nodes we can write the vector
form of the previous equation as follows:
\[
	\Expec{}{\bx^{(t)} \;|\; \bx^{(t-1)}=\bx} 
	= \frac{\alpha}{n}\bone + \frac{1}{n}\Big[ (1 - \alpha)P + (n-1)I\Big]\bx.
\] 

This immediately implies the following equation, relating expected states at
times $t - 1$ and $t$ (with $\Expec{}{\bx^{(0)}} = \bx$):
\begin{align*}
& \Expec{}{\bx^{(t)}} 
= \frac{\alpha}{n}\bone + \frac{1}{n}\Big[ (1 - \alpha)P + (n-1)I\Big] \Expec{}{\bx^{(t-1)}}.
\end{align*}

Now, consider $\bone - \bx^{(t)}$, the difference between the absorbing state
vector $\bone$ and the state vector at a generic time $t$. Obviously, $(\bone -
\bx^{(t)})_v \ge 0$ deterministically, for every $v$ and for every $t$. As for
the expectation of this difference, we have:

{\small\begin{align}\label{eq:expec_diff}
&\Expec{}{\bone - \bx^{(t)}} 
= \frac{1}{n}\Big[ (1 - \alpha)P + (n-1)I\Big] \Expec{}{\bone - 
\bx^{(t-1)}},
\end{align}}
where the equality follows by collecting and rearranging terms, after observing
that both matrices $P$ and $I$ have eigenvalue $1$ with associated eigenvector
$\bone$. Moreover, we have
\[
	\frac{1}{n}\Big[ (1 - \alpha)P + (n-1)I\Big] = \left(1 - \frac{\alpha}{n}\right)\hat{P},
\]
with $\hat{P} := \frac{n-1}{n-\alpha}\left[\left(\frac{1-\alpha}{n-1}\right)P +
I\right]$ a stochastic matrix. This follows immediately by observing that both
$P$ and $I$ are stochastic, so that all rows of $(1 - \alpha)P + (n-1)I$
identically sum to $n - \alpha$. By solving the recursion in
Eq.~\eqref{eq:expec_diff} we obtain
\[
	\Expec{}{\bone - \bx^{(t)}} = 
	\left(1 - \frac{\alpha}{n}\right)^t\hat{P}^t \left[	\bone-\bx^{(0)}\right] 
	\stackrel{(a)}{=} \left(1 - \frac{\alpha}{n}\right)^t\bone\ - 
	\left(1 - \frac{\alpha}{n}\right)^t\hat{P}^t \bx^{(0)},
\]
where in $(a)$ we use the fact that $\hat{P}^t$ is a stochastic matrix, thus
with main eigenvalue $1$ and associated eigenvector $\bone$.  Next, observe
that for every $v$, we have $\left(\hat{P}^t \bx^{(0)}\right)_v\ge 0$, so we
also have $\Expec{}{1 - x_v^{(t)}}\le\left(1 - \frac{\alpha}{n}\right)^t$.

Therefore, for every time $t \geq \frac{2}{\alpha} n\log n$ we have
\[
	\Expec{}{1 - x_v^{(t)}} \leq \frac{1}{n^2}
\]
for every $v \in V$. Because the $x_v^{(t)}$'s are binary random variables
\[
\Prob{}{x_v^{(t)} = 0}
= \Prob{}{1 - x_v^{(t)} = 1} \leq \Prob{}{1 - x_v^{(t)} \geq 1} \\
\leq \Expec{}{1 - x_v^{(t)}} \leq \frac{1}{n^2}
\]
where in the second-to-last inequality we used the Markov inequality. Hence, in
$\mathcal{O}(\frac{1}{\alpha}n\log n)$ rounds the process converges to the
absorbing state $\bone$, with high probability.
\end{proof}

Note that Theorem~\ref{thm:conv_time_asy} implies that the convergence time is
still $\mathcal{O}\left(n^{1+s}\log n\right)$ when $\alpha =
\Theta\left(\frac{1}{n^s}\right)$ for any $s > 0$, hence polynomial as long as
$s$ is constant.

%% file: trunk/conclusions.tex
\section{Discussion and Outlook}\label{se:concl}
In this paper, we considered biased opinion dynamics under two popular update
rules, namely majority~\cite{krapivsky2003dynamics} and the voter
model~\cite{liggett2012interacting}. Although related, these two models exhibit
substantial differences in our setting. Whereas the voter model enforces a
drift toward the majority opinion within a neighborhood, in the sense that this
is adopted with probability proportional to the size of its support, majority
is a nonlinear update rule, a feature that seems to play a crucial role in the
scenario we consider. This is reflected in the absorption time of the resulting
biased opinion dynamics, which is $\bigO\left(\frac{1}{\alpha}n \log n\right)$
for the voter model, regardless of the underlying topology, whereas it exhibits
a far richer behavior under the majority rule, being super-polynomial (possibly
exponential) in dense graphs.  It may be worth mentioning that in the case of
two opinions, the majority rule is actually equivalent to the (unweighted)
median rule, recently proposed as a credible alternative to the weighted
averaging of the DeGroot's and Friedkin-Johnsen's models~\cite{mei2019occam}.

\paragraph{A modular model.} Both scenarios we studied are instantiations of a
general model that is completely specified by a triple $(\bz, \alpha,
\mathcal{D})$, with $\bz$ an initial opinion distribution, $\alpha \in (0,1]$ a
probability measuring the magnitude of the bias toward the \better\ opinion,
and $\mathcal{D}$ an {\em update rule} that specifies some underlying opinion
dynamics. In more detail, a biased opinion dynamics can be succinctly described
as follows.

\smallskip The system starts in some state $\bx^{(0)}= \bz$, corresponding to
the initial opinion distribution; for $t>0$, let $\bx^{(t-1)} = \bx$ be the
state at the end of step $t - 1$. In step $t$, a node $v$ is picked uniformly
at random from $V$ and its state is updated as follows:
\[
	x_v^{(t)} = \left\{\begin{array}{ll}
		1			& \text{with probability }	\alpha,
		\\
		\mathcal{D}_G(v, \bx)	& \text{with probability } 	1-\alpha,
	\end{array}\right.
\]
where $\mathcal{D}_G: V\times\{0, 1\}^n \rightarrow \{0,1\}$ is the update
rule.\footnote{The subscript $G$ highlights the fact that result of the
application of a given update rule $\mathcal{D}$ in general depends on both the
current state and the underlying graph $G$. The above definition can be easily
adjusted to reflect the presence of weights on the edges.} When the update rule
is probabilistic (as in the voter model), $\mathcal{D}_G(v, \bx)$ is a random
variable, conditioned to the value $\bx$ of the state at the end of step $t -
1$.

\noindent\emph{Remark.} It is simple to see that $\bone$ is the only absorbing
state of the resulting dynamics, whenever $\alpha \neq 0$ and $\mathcal D$ does
not allow update of an agent's opinion to one that is not held by at least one
of  the agent's neighbors, which is the case for many update rules in the
discrete-opinion setting.

\smallskip We further remark that the initial condition $\bx^{(0)} = \bzero$
considered in this paper is not intrinsic to the model, it rather reflects
scenarios (e.g., technology adoption) where a new, superior alternative to the
status quo is introduced, but its adoption is possibly slowed by inertia of the
system. Although the reasons behind system's inertia are not the focus of this
paper, inertia itself is expressed here as a {\em social pressure} in the form
of some update rule $\mathcal{D}_G$. Another reason for choosing a fixed
initial state ($\bzero$ in our case) is being able to compare the behavior of
the biased opinion dynamics under different update rules on a common basis. 

Finally, it is worth mentioning that Theorem~\ref{thm:conv_time_asy} and the
upper bounds given in Section~\ref{subse:other} hold regardless of the initial
opinion distribution.

\paragraph{Outlook.} This paper leaves a number of open questions. A first one
concerns general upper bounds on convergence times under the majority update
rule. Even though the topology-specific upper bounds given in
Section~\ref{sec:majority} might suggest general upper bounds that depend on
the maximum degree, thus mirroring the result of Theorem~\ref{thm:min_degree},
this turns out to not be the case, with preliminary experimental results
suggesting a more complicated dependence on degree distribution.  In
particular, convergence time is fast for reasonably large values of
$\alpha<1/2$ (e.g., $\alpha \geq 0.1$), but it presents a sharp transition as
$\alpha$ becomes sufficiently small (e.g., this happens around $0.01$ for the
Petersen graph).  The dependence of this threshold on $n$ is not completely
clear.

A further question is whether expected absorption time is always
$\mathcal{O}(n\log n)$ when $\alpha \geq 1/2$, irrespective of the underlying
dynamics and topology. This is clearly true for the voter model from
Theorem~\ref{thm:conv_time_asy} and it also holds for majority, whenever the
underlying network has minimum degree $\Omega(\log n)$ (see comment after
Lemma~\ref{lemma:biasedbdchainsprob} for an idea of the proof argument for
dense graphs).

We finally remark that our results and most results in related work apply to
the case of two competing opinions. An obvious direction for further research
is extending our results to the case of multiple opinions.